\documentclass[dvips]{jmi0}
\usepackage[pdftex]{graphicx}
\title{A complete transformation rule set and a minimal equation set \\for CNOT-based 3-qubit quantum circuits (Draft)}
\author{IS}{Issei Sakashita}{i-sakashita@math.kyushu-u.ac.jp}
\affiliation{IS}{Kyushu University, 744 Motooka, Nishi-ku, Fukuoka 819-0395, Japan}%
\abstract{%
We introduce a complete transformation rule set and a minimal equation set 
for controlled-NOT (CNOT)-based quantum circuits.
Using these rules, quantum circuits that compute the same Boolean function are reduced to a same normal form.
We can thus easily check the equivalence of circuits by comparing their normal forms. By applying the Knuth-Bendix completion algorithm to a set of modified 18 equations introduced by Iwama et al. 2002 \cite{Iwama}, we obtain a complete transformation rule set (i.e., a set of transformation rules with the properties of `termination' and `confluence').
Our transformation rule set consists of 114 rules.
Moreover, we found a minimal subset of equations for the initial equation set. 
}
\usepackage{slashbox}
\begin{document}
\maketitle
%%%%%%%%%%%%%%%%%%%%%%
%
%
%    Intorduction
%
%
%%%%%%%%%%%%%%%%%%%%%%
\section{Introduction}
Quantum computers were proposed in the early 1980s \cite{benioff80, benioff}.
Significant contributions to quantum algorithms include the Shor factorization algorithm \cite{Shor94, Shor97} and the Grover search algorithm \cite{grover96}.
The quantum-circuit model of computation is due to Deutsch \cite{deu89}, and it was further developed by Yao \cite{Yao93}.

After the works of Deutsch and Yao the concept of a universal set of quantum gates became central in the theory of quantum computation. 
A set $G = \{G_{1,n_1},\cdots , G_{r,n_r}\}$
 of $r$ quantum gates $G_{j,n_j} $ acting on $n_j$ qubits $(j = 1,\cdots , r)$, is called universal if any unitary action $U_n$ on $n$ input quantum states can be decomposed into a product of successive actions of $G_{j,n_j}$ on different subsets of the input qubits \cite{Mar}.
 A first example of 3 qubit universal gate set consists of Deutsch's gates $\mathbf{Q}$ \cite{deu89}. 
 The gate $\mathbf{Q}$ is an extension of the Toffoli gate \cite{Toffoli81}. 
DiVincenzo showed that a set of two-qubit gates is exactly universal for quantum computation \cite{div95}.
After the result of DiVincenzo, Barenco showed that a large subclass of two-qubit gates are universal, and moreover, that almost any two-qubit gate is universal \cite{Barenco95A}. 
Barenco et al. showed that the set consisting of one-qubit gates and CNOT gates is universal \cite{Barenco95B}. 
There have been a number of studies that investigate the number of gates for decomposing an any gate of $n$ qubits in $U(2^n)$.
For the universal set consisting of one-qubit gates and CNOT gates, the number of gates is $O(n^34^n)$ by Barenco et al. \cite{Barenco95B}.
Knill reduced this bound to $O(n4^n)$ \cite{Knill95}.
Most useful information about universal quantum gates can be obtained from a survey paper written by A. Galindo and M.A. Marin-Delgado \cite{Mar}.

The design of a good quantum circuit plays a key role in the successful implementation of a quantum algorithm.
For this reason, Iwama et al. presented transformation rules that transform any `proper' quantum circuit into a `canonical' form circuit \cite{Iwama}.
There is, however, no discussion about the minimal size of a quantum circuit.
In this article, we formulate a quantum circuit as a string and then simplify the circuit by using string rewriting rules to investigate them formally. 
Since a string rewriting system can be analyzed by using a monoid, we require several properties about monoids and groups. 

String rewriting systems simplify strings by using transformation rules, and they have played a major role in the development of theoretical computer science. Several studies of string rewriting systems have been investigated \cite{Book}.
Let $M$ be a monoid and $T$ a submonoid of finite index in $M$.
If $T$ can be presented by a finite complete rewriting system, so $M$ can \cite{JW97}.
The problem of confluence is, in general, undecidable.
Parkes et al. showed that the class of groups that have monoid presentations obtainable by finite special $[\lambda]$-confluent string rewriting systems strictly contains the class of plain groups \cite{parkes2004}.
%There is no method to solve the word problem in general form.
%It is known that  in a number of specific domains such as finitely presented groups and monoids the word problem is undecidable. 
%If $R$ is a finite string rewriting system on alphabet $\Sigma$, then the word problem for $R$ is decidable.
The word problem for a finite string is, in general, undecidable. 
If R is a finite string rewriting system that are Noetherian and confluent, then the word problem is decidable \cite{Boo82a, Otto91}.
Book considered the word problem for finite string rewriting systems in which the notion of `reduction' is based on rewriting the string as a shorter string \cite{Boo82a}. 
He showed that for any confluent system of this type, there is a linear-time algorithm for solving the word problem.
% the bound on the linear time depends very heavily on the fact that the notion of `reduction' is based on rewriting the strings as shorter strings.
Using a technique developed in \cite{Boo82a}, Book and \'O'D\'ulaing \cite{BoODu81} showed that there is a polynomial-time algorithm for testing if a finite string rewriting system is confluent. 
Gilman \cite{Gil79} considered a procedure that, beginning with a finite string rewriting system, attempts to construct an equivalent string rewriting system that is Noetherian and  confluent, that is, a string rewriting system such that every congruence class has a unique `irreducible' string.
This procedure appears to be a modification of the completion procedure developed by Knuth and Bendix \cite{KnBe70} in the setting of term-rewriting systems. Narendran and Otto \cite{NaOt88a} 
also contributed to this topic.
Later, Kapur and Narendran \cite{kana85b} showed how the Knuth-Bendix completion procedure could be adapted to the setting of string rewriting systems.

We do not deal with the general theory whether string rewriting systems are decidable or undecidable.
We introduce an idea to reduce the size of a quantum circuit by using a string rewriting system. 
Our string rewriting rules based on 18 equations introduced by Iwama et al. 2002 \cite{Iwama}. 
The Iwama's equations can not be considered as a complete rewriting rules as it is.
That is, it does not have properties of termination and confluence.
We would like to obtain a complete transformation rule set (i.e., a set of transformation rules with the properties of `termination' and `confluence') for reducing a quantum circuits. 
Therefore, we apply the Knuth-Bendix completion algorithm to a modified 18 equations. 
We obtain our complete transformation rule set consisted of 114 rules.
%And we investigate what kind of rewriting rule is needed to reduce the size of a  quantum circuit. 
There are three major results that are obtained by our investigation, for 3 qubits, the length of normal form is at most 6, the number of normal form is 168.
Furthermore, we found a minimal subset of equations.
The number of the general quantum circuits to arbitrary $n$ qubits is already known by research of Clifford groups. 
So $CQC_3$ is considerd as a subgroup of a Clifford group.

This article consists of as follows.
In section 2, we describe formal definitions  of a quantum circuit.
We consider a circuit that consists of just CNOT gates on 3 qubits. 
In section 3, we prove several properties for string rewriting systems.
In section 4, we define a quantum circuit rewriting system for 3 qubits and show several related properties about it. We show that the number of normal forms is 168 on 3 qubits.
In section 5, we found a minimal subset of equations.
%Finally, for any size $n$, we calculate the number of  $n$ quits quantum gates and equations to formalize $n$ qubits quantum circuit rewriting system.
%%%%%%%%%%%%%%%%%%%%%%%%%%%%
%
%
%
%%%%%%%%%%%%%%%%%%%%%%%%%%%%
\section{Definitions of Quantum Circuits}
In this section, we introduce several definitions related to quantum circuits.
First, we define quantum bits (qubits), quantum gates, and quantum circuits.
\begin{definition}[Quantum bits, gates, and circuits]
Let $\alpha , \beta \in \mathbb{C}, | 0 \rangle=(1,0), | 1 \rangle=(0,1)$ and $m \in \mathbb{N}$.
\begin{itemize}
\item A single qubit is denoted by a vector $| x \rangle = \alpha  | 0 \rangle + \beta | 1 \rangle$.
\item A $n$-qubit is denoted by $| x_1 \rangle \otimes | x_2 \rangle \otimes \cdots \otimes | x_n \rangle \in \mathbb{C}^{2^n}$.
\item A $n$-qubit quantum gate is an unitary operator  \\ $G : \mathbb{C}^{2^n} \to \mathbb{C}^{2^n}$.
\item A quantum circuit $Cir$ of size $m$ is denoted by $Cir  = (G_1,G_2,\cdots , G_m)$ where $G_i$  ($i = 1,2,\cdots ,m$) are $n$-qubit quantum gates. 
\item An empty circuit is denoted by $\lambda$.
\item The output of  circuit $Cir = (G_1,G_2,\cdots , G_m)$ for an input $|x \rangle$ is $(G_m \circ \cdots \circ G_2 \circ G_1) |x \rangle$.
\end{itemize}
\begin{definition} Let $m, l \in N$, $Cir_1= (G_1,G_2,\cdots , G_m)$ and $Cir_2 = (G'_1,G'_2,\cdots , G'_l)$ be $n$-qubit quantum circuits.
We define an equivalence relation $=_{cir}$ by 
\begin{eqnarray*}
&&Cir_1 =_{cir}Cir_2\\
&&\iff \forall |x\rangle \in C^2 , (G_m \circ \cdots \circ G_1) |x\rangle = (G'_l \circ \cdots  \circ G'_1) | x\rangle.
\end{eqnarray*}
\end{definition}
Next, we introduce a quantum gate that plays an important role in proving the universality of a quantum circuit. 
\begin{definition}
The $n$-qubit controlled-NOT (CNOT) gate is a unitary operator  $[c ,t]_n : \mathbb{C}^{2^n} \to \mathbb{C}^{2^n}$ 
($c , t \in \{1,2, \cdots ,n\}$) defined by
$$\bigotimes_{i=1}^{n} |\delta_i \rangle \mapsto 
\bigotimes_{i=1}^{t-1} |\delta_i \rangle 
\otimes |\delta_t \oplus \delta_c \rangle \otimes  
 \bigotimes_{i=t+1}^{n} |\delta_i \rangle.$$
We call $c$ the \textit{control bit} and $t$ the \textit{target bit}.
We use a version of Feynmann's notation \cite{Feynmann85} for diagrammatic representations of CNOT gates (cf. Figure \ref{cnot}).
An example of 3 qubit quantum circuit is illustrated in Figure \ref{circuit01}. Each gate is applied in turn from left to right to the $n$ qubits.
\begin{figure}[h]
\begin{center}
  \includegraphics[width=30mm]{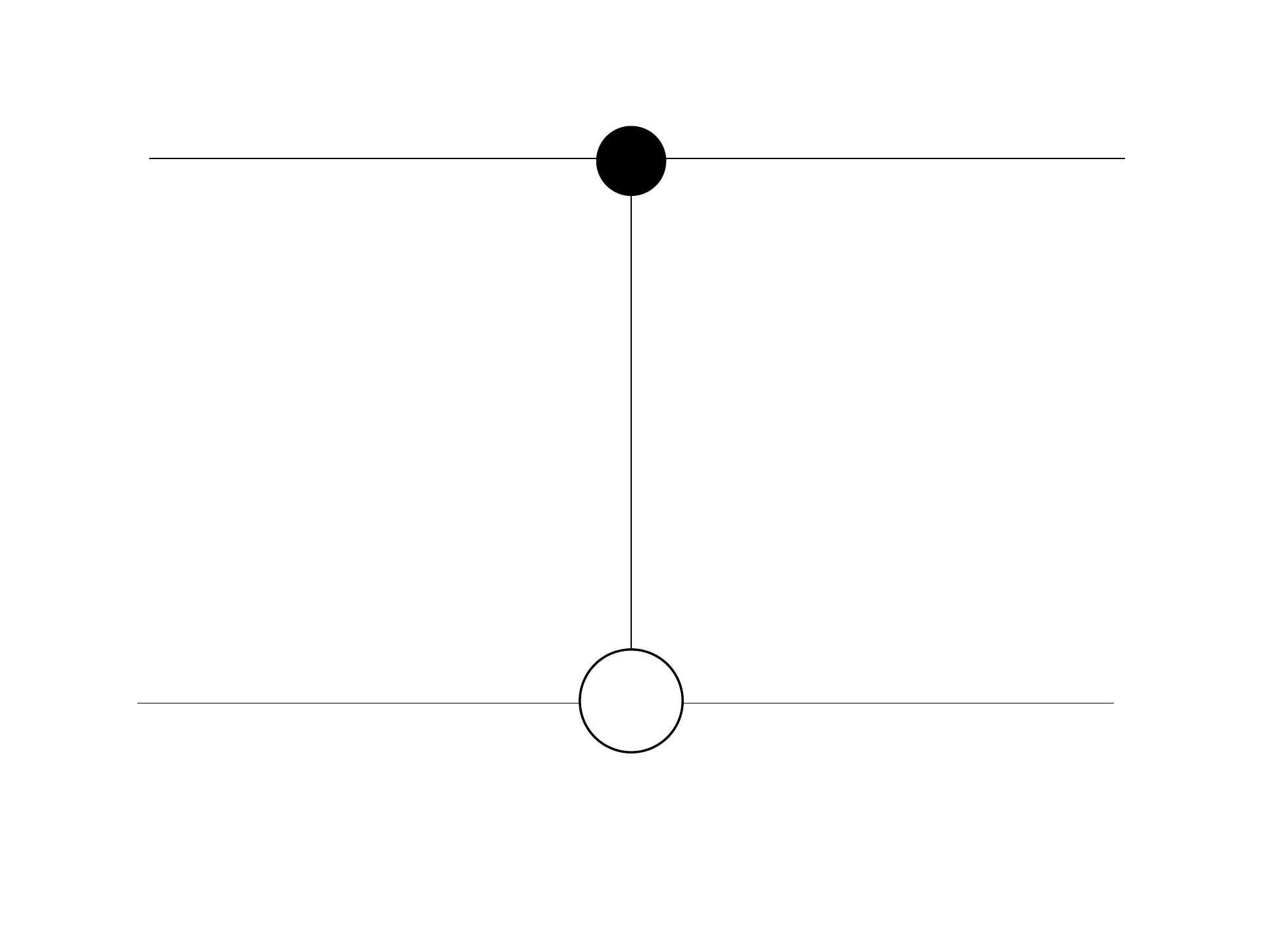}
  \caption{2-qubit CNOT gate}\label{cnot}
\end{center}
\end{figure}
\end{definition}
\begin{figure}[h]
\begin{center}
  \includegraphics[width=50mm]{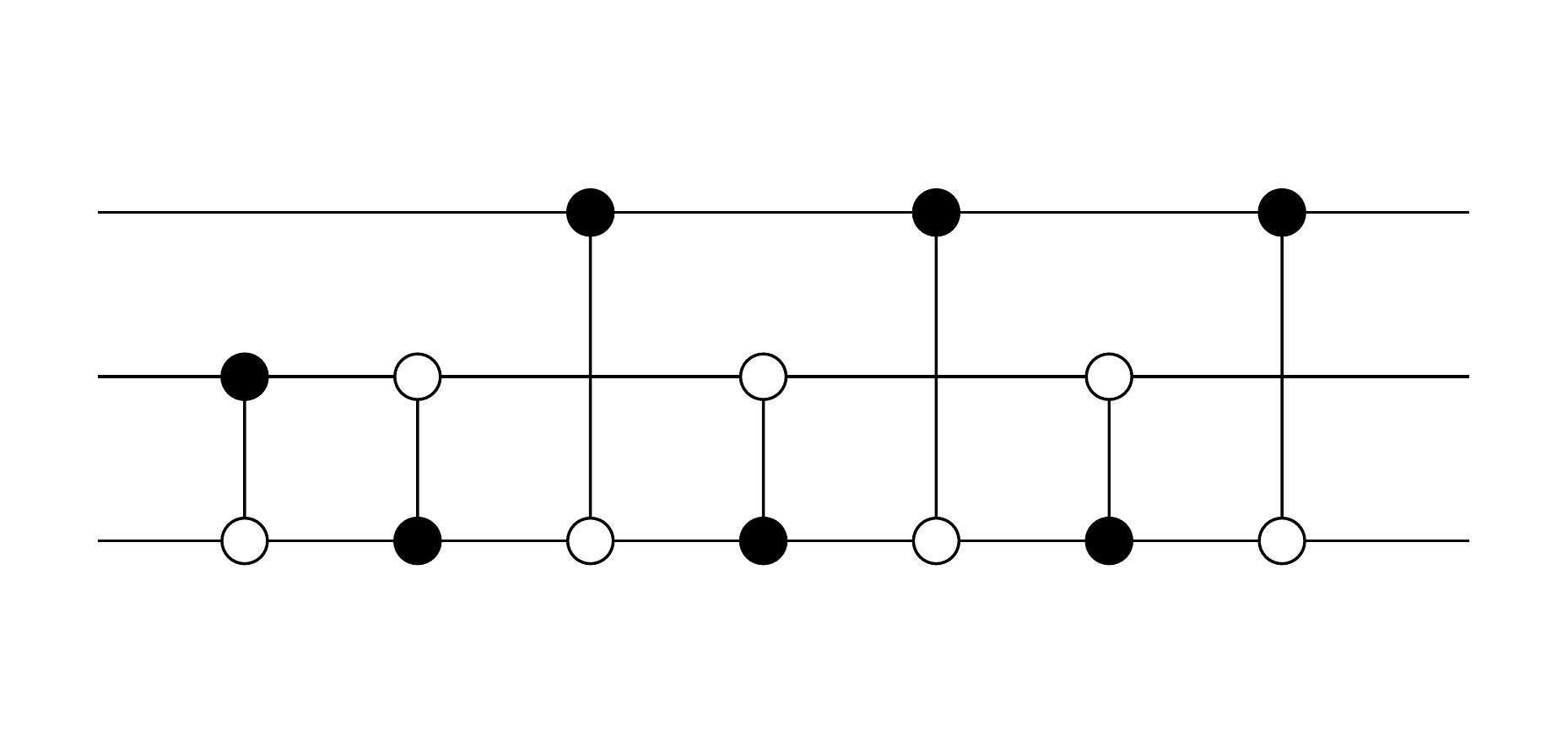}
  \caption{A quantum circuit}\label{circuit01}
\end{center}
\end{figure}
\end{definition}
Next, we define an equivalence relation between two quantum circuits. 
This definition is important and allows us to discuss the equivalence of  circuits.
In this paper, we consider only quantum circuits that are constructed by 3-qubit CNOT gates. We denote as the set of circuits $CQC_3$ as
$$CQC_3 =
 \{([c_1,t_1]_3 ,[c_2,t_2]_3,\cdots ,[c_m,t_m]_3) \vert$$
 $$ c_i,t_i \in \{1,2,3\}, c_i \neq t_i , m \in \mathbb{N}\}.$$
We note that two different circuits $Cir_1$ and $Cir_2$ in  $CQC_3$ , may be equivalent in the sense of $=_{cir}$, i.e. $Cir_1 =_{cir} Cir_2$.
\begin{example}
The following equation can be considered as illustrated in Figure.\ref{circuit02}: 
\begin{equation*}
([1, 2]_3,[2,3]_3) =_{cir} ([2,3]_3,[1,2]_3,[1,3]_3).
\end{equation*}
For all input qubits $|\delta_1 \rangle \otimes |\delta_2 \rangle \otimes |\delta_3 \rangle$,
we need to prove the equivalence of the outputs.
First, we compute $([1,3]_3 \circ [1,2]_3 \circ [2,3]_3) (|\delta_1 \rangle \otimes |\delta_2 \rangle \otimes |\delta_3 \rangle)$,
\begin{eqnarray*}
&&([2,3]_3 \circ [1, 2]_3) (|\delta_1 \rangle \otimes |\delta_2 \rangle \otimes |\delta_3 \rangle)  \\
&=_{cir}& [2,3]_3 (|\delta_1 \rangle \otimes |\delta_1 \oplus \delta_2 \rangle \otimes |\delta_3 \rangle) \\
&=_{cir}& |\delta_1 \rangle \otimes |\delta_1 \oplus \delta_2 \rangle \otimes |\delta_1 \oplus \delta_2 \oplus \delta_3 \rangle. 
\end{eqnarray*}
Next, we compute $([1,3]_3 \circ [1,2]_3 \circ [2,3]_3) (|\delta_1 \rangle \otimes |\delta_2 \rangle \otimes |\delta_3 \rangle)$,
\begin{eqnarray*}
&&([1,3]_3 \circ [1,2]_3 \circ [2,3]_3) (|\delta_1 \rangle \otimes |\delta_2 \rangle \otimes |\delta_3 \rangle)   \\
&=_{cir}&([1,3]_3 \circ [1,2]_3) (|\delta_1 \rangle \otimes |\delta_2 \rangle \otimes |\delta_2 \oplus \delta_3 \rangle)  \\
&=_{cir}&[1,3]_3 (|\delta_1 \rangle \otimes |\delta_1 \oplus \delta_2 \rangle \otimes |\delta_2 \oplus \delta_3 \rangle)  \\
&=_{cir}& |\delta_1 \rangle \otimes |\delta_1 \oplus \delta_2 \rangle \otimes |\delta_1 \oplus \delta_2 \oplus \delta_3 \rangle. 
\end{eqnarray*}
Thus we have $([1, 2]_3,[2,3]_3) =_{cir} ([2,3]_3,[1,2]_3,[1,3]_3)$.
\begin{figure}[h]
\begin{center}
  \includegraphics[width=50mm]{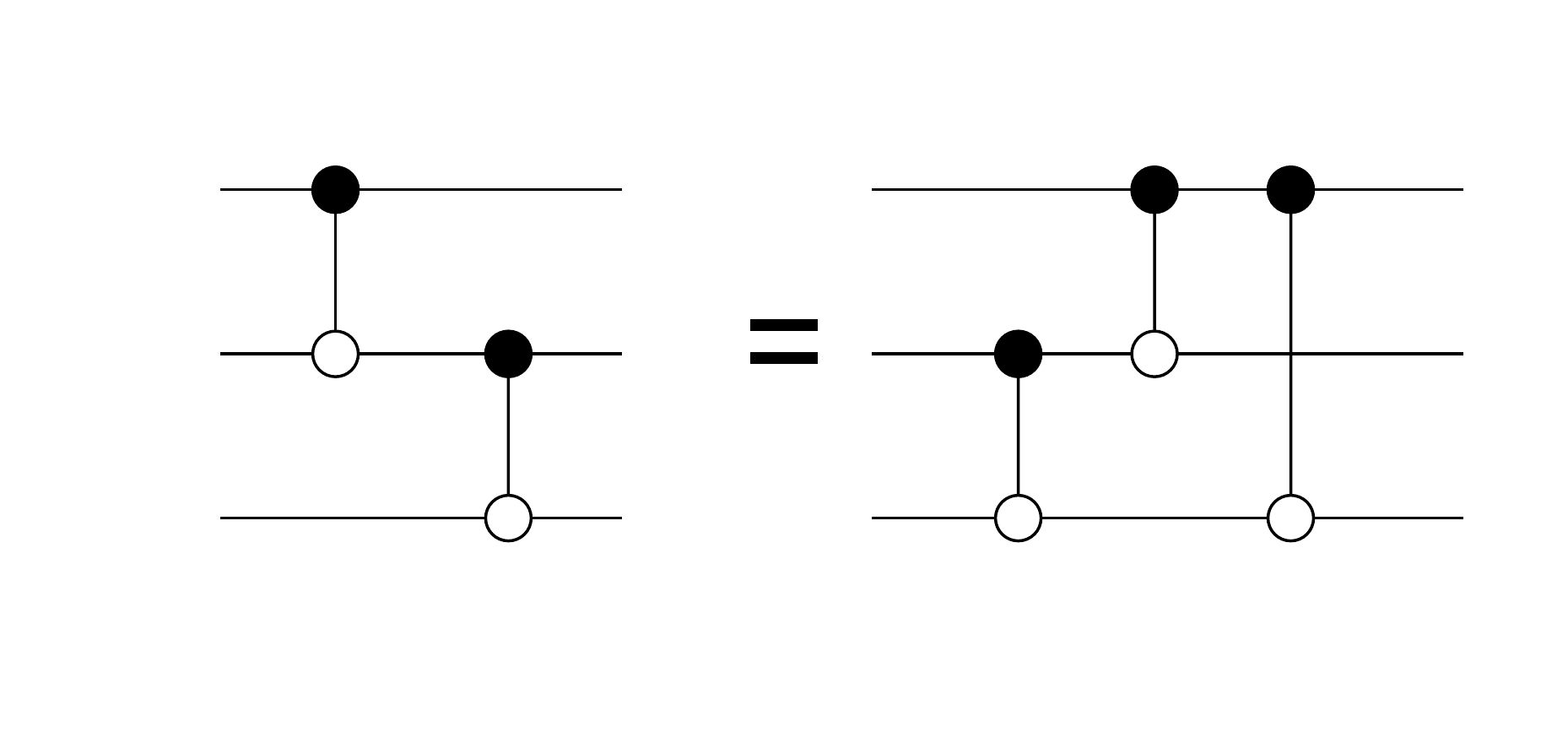}
  \caption{A circuit equation}\label{circuit02}
\end{center}
\end{figure}
\end{example}
We chose three types of simple equations to construct  a string rewriting system. 
\begin{definition} Let $G_1, G_2$, and $G_3 \in CQC_3$ be CNOT gates.
\begin{itemize}
\item For any CNOT gate $G$, $(G,G) =_{cir} \lambda$ is an eliminated equation. \item $(G_1,G_2) =_{cir} (G_2 ,G_1)$ is a commutative equation. 
\item $(G_1,G_2) =_{cir} (G_2 ,G_1,G_3)$ is a anti-commutative equation . \end{itemize}
In this article, we denote six CNOT gates for the 3 qubits $a=[1,2]_3$, $b=[1,3]_3$, $c=[2,1]_3$, $d=[2,3]_3$, $e=[3,1]_3$ and $f=[3,2]_3.$
\begin{figure}[h]
\begin{center}
  \includegraphics[width=40mm]{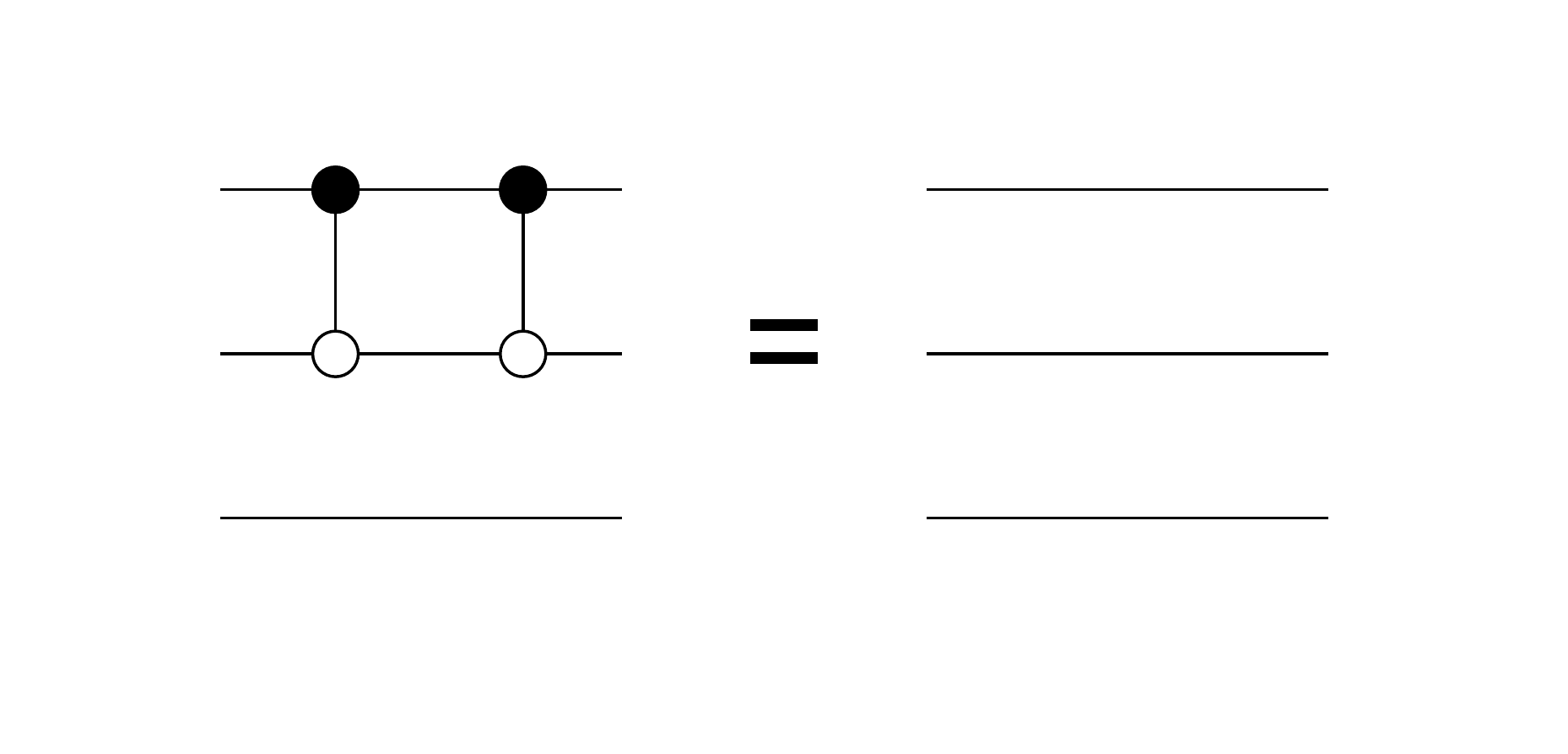}
  \caption{eliminated type: (a , a) = $\lambda$}
\end{center}
\end{figure}
\begin{figure}[h]
\begin{center}
  \includegraphics[width=40mm]{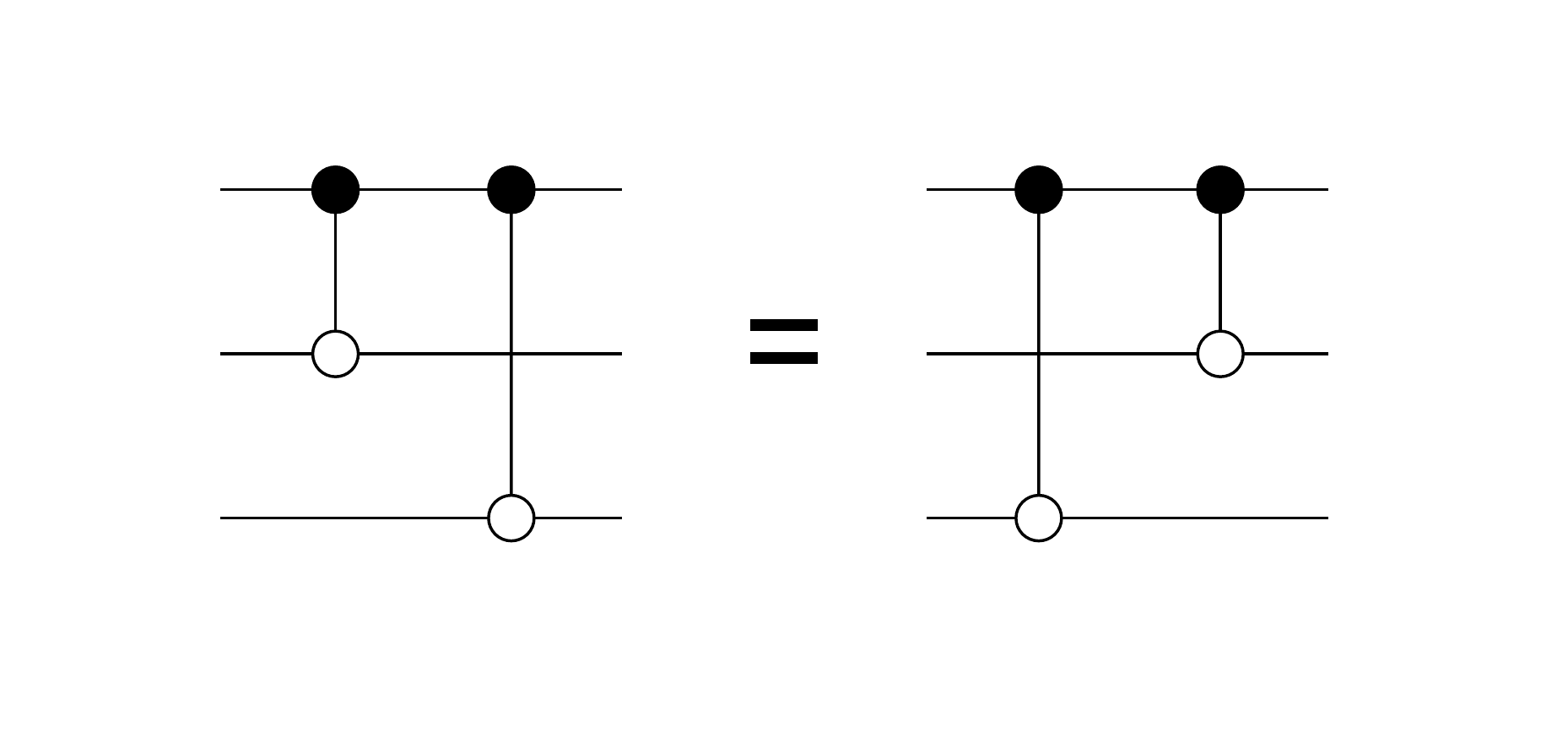}
  \caption{commutative type: (a , b)  = (b , a)}
\end{center}
\end{figure}
\begin{figure}[h]
\begin{center}
  \includegraphics[width=40mm]{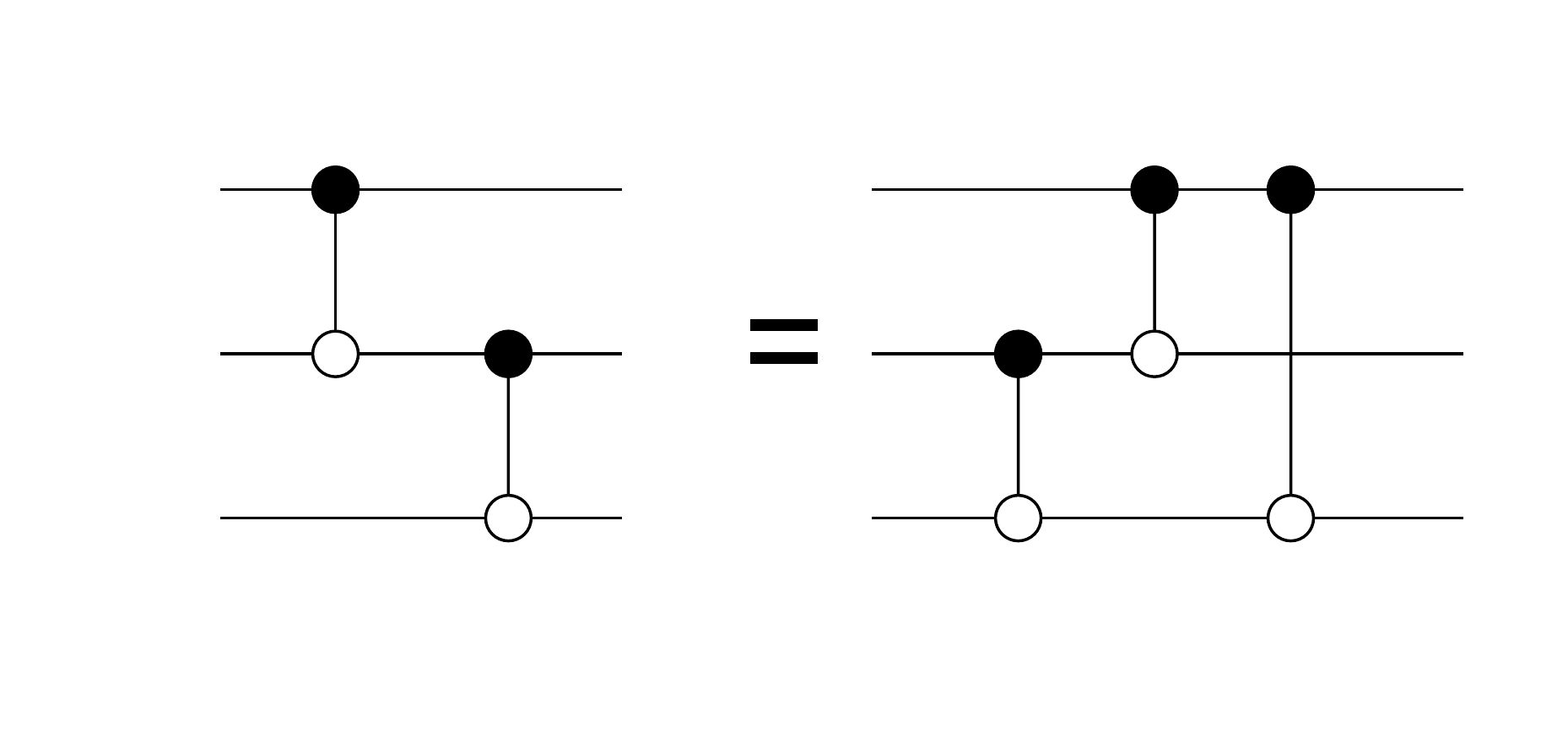}
  \caption{anti-commutative type: $(a, d) =_{cir} (d , a , b)$}
\end{center}
\end{figure}
\end{definition}
%%%%%%%%%%%%%%%%%%%%%
%
%
%%%%%%%%%%%%%%%%%%%%%
\section{String Rewriting System}
In this section, we introduce the definition of a string rewriting system, in order to discuss quantum circuit transformation systems using it.
Let $\Sigma$ be a finite set of alphabets. We denote the set of all strings over $\Sigma$, including the
empty string $\lambda$, as $\Sigma^*$.
The length of a string $w \in \Sigma^*$ is denoted by $|w|$.
A rewriting rule $(u, v)$ is a pair of strings $u, v \in \Sigma^*$ where $u \neq \lambda$.
\begin{definition}[string rewriting system]
A string rewriting system is a pair $(\Sigma, R)$ of a finite set of alphabet $\Sigma$ and a finite set of rewriting rules $R$ .
\end{definition}
\begin{definition}[string rewriting]
Let $(\Sigma , R)$ be a rewriting system and $s, t \in \Sigma^*$.
We denote $s \to_R t$ if and only if there exist strings $x, y , u$ 
and $ v$ in $\Sigma^*$ such that $s = x u y, t = x v y$ and $ (u, v)\in R$.  \\
The reflexive transitive closure relation of $\to_R$ over $\Sigma^*$ is denoted by $\to_{R}^{*}$. Further $\leftrightarrow _{R}^{*}$ is the symmetric closure relation of $\to_{R}^{*}$.
\end{definition}
\begin{definition}[irreducible, normal form]
Let ($\Sigma$, $R$) be a rewriting system and $w \in \Sigma^*$.
For all substrings $c \subset w$, if there are no rules $(c,c') \in R$, then $w$ is $irreducible$. 
For $s \in \Sigma^*$, if there exists $s'$ such that $s \to ^*_R s'$ and $s'$ is irreducible, then $s'$ is the $normal form$ of $s$.
We denote the normal form $s'$ of $s$  as $NF(s)$.
\end{definition}
The equivalence class of a string rewriting systems are
considered using monoids, so we introduce several definitions and properties about monoids and their interprelations.
\begin{definition}[monoid]
A monoid $M = (M, \cdot, \lambda)$ is a tuple of a set $M$, a binary operation $\cdot : M \times M \to M$, and a unit element $e \in M$
 that satisfies the following two axioms.
 \begin{itemize}
\item For any $a, b$, and $c$ in $M, (a \cdot b) \cdot c = a \cdot (b \cdot c)$.
\item For any $a$ in $M$, $a \cdot \lambda = \lambda \cdot a = a$.
\end{itemize}
We note $(\Sigma ^*, \cdot , \lambda)$ is a monoid where $\cdot$ is concatenation and $\lambda$ is an empty string.
\end{definition}
\begin{definition}[homomorphism, isomorphic]A homomorphism between two monoids $(M_1, \cdot_1, \lambda_1)$ and $(M_2, \cdot_2, \lambda_2)$ is a function $f : M_1 \to M_2$ such that
\begin{itemize}
\item $f(x \cdot_1 y) = f(x) \cdot_2 f(y)$ for any $x , y \in M_1$, and
\item $f(\lambda_1) = \lambda_2$.
\end{itemize}
If there exists a bijective homomorphism $f : M_1 \to M_2$, then $M_1$ and $M_2$ are isomorphic. We denote isomorphic as $M_1 \sim M_2$.
\end{definition}
\begin{proposition} Let $\Sigma$ be a finite set and $(M,\cdot ,\lambda)$ a monoid. A function $f : \Sigma \to M$ is uniquely extended to the homomorphism $f^* : \Sigma^* \to M$ where $f^{*}(x_1\cdot x_2 \cdots x_n) = f(x_1) f(x_2) \cdots f(x_n)$ and
$f^{*}(\lambda) = \lambda$.
\end{proposition}
\hfill $\qed$
\begin{definition}[model, interpretation]
Let $(\Sigma, R)$ be a rewriting system and $(M, \cdot , \lambda)$ a monoid.
We say $(M, \cdot , \lambda)$ is a model of $(\Sigma, R)$ if there exists a function $f : \Sigma \to M$ such that $f^*(u) = f^*(v)$ for any $(u, v) \in R$.  We call the function $f^*$ an interpretation of the string rewriting system $(\Sigma, R)$ to a monoid $M$. 
\end{definition}
A string rewriting system can be investigated using a monoid and interpretation.
We now add further definitions for discussing the equivalence of rewriting systems.
\begin{definition}[factor monoid]
Let $(\Sigma, R)$ be a rewriting system. A factor monoid $(\Sigma ^* / R, \cdot , [\lambda])$ is defined by 
$\Sigma^* / R = \Sigma^* / \leftrightarrow^{*}_{R}$ and  $[x] \cdot [y] = [x y]$ where $[x] = \{x' | x \leftrightarrow_{R}^{*} x' \}$
\end{definition}
\begin{proposition}
Let $(\Sigma, R)$ be a rewriting system, $(M, \cdot, \lambda)$ a model of $R$ and $f^* : \Sigma^{*} \to M$ an interpretation. 
The function $[f^*] : \Sigma^* /R \to M$ defined by $[f^*]([x]) = [f^*(x)] (x \in \Sigma^*)$ is a homomorphism.
\hfill $\qed$
\end{proposition}
\begin{definition}[rewriting system equivalence]
Let $(\Sigma, R_1)$ and $(\Sigma, R_2)$ be rewriting systems. $R_1$ and $R_2$ are equivalent if and only if $\Sigma^* /R_1$ and $\Sigma^* /R_2$ are isomorphic. 
\end{definition}
Finally, we introduce a lemma to compare two rewriting system that have the same alphabet $\Sigma$.
\begin{lemma}\label{interpretation}
Let $(\Sigma, R_1)$ and $(\Sigma, R_2)$ be rewriting systems, and let $(M, \cdot, \lambda)$ be a monoid of $(\Sigma, R_2)$.
If there exists $(x_1, x_2) \in R_1$ and an interpretation $f : \Sigma^*  \to M$ for $\Sigma^* /R_2$ such that $f^*(x_1) \neq f^*(x_2)$,
then $\Sigma^* / R_1 \nsim \Sigma^* /R_2$.
\end{lemma}
\hfill $\qed$
%%%%%%%%%%%%%%%%%%%%%%%%
%
%
%
%
%%%%%%%%%%%%%%%%%%%%%%%%
\section{Quantum circuit rewriting system}
We define a quantum circuit rewriting system for $CQC_3$.
\begin{definition}[Quantum circuit rewriting system]
Let $(\Sigma,R)$ be a string rewriting system and $i^* : \Sigma^* \to CQC_3/=_{cir}$ a function where
$\Sigma=\{a,b,c,d,e,f\}$, $i(a) = [1,2]_3$, $i(b) = [1,3]_3$, $i(c) = [2,1]_3$, $i(d) = [2,3]_3$, $i(e) = [3,1]_3$ and $i(f) = [3,2]_3$.
$(\Sigma, R)$ is a quantum circuit rewriting system, if $i^*$ is an interpretation of $(\Sigma, R)$.
We identify a string $w = x_1x_2\cdots x_n \in \Sigma^*$ as a circuit $(i(x_1), i(x_2), \cdots , i(x_n)) \in CQC_3$, and we also call $w$ a $circuit$.
\end{definition}
In general, string rewriting system do not have properties of `termination' and `confluence'.
So we would like to construct a quantum circuit rewriting system that has both properties termination and confluence.
To do so, we use the Knuth-Bendix algorithm \cite{KnBe70, Book, YM83}.
\begin{definition}
Let  $E$ be an equation set.
If the Knuth-Bendix algorithm succeeds for $E$, then
we have a complete transformation rule set $R$ (i.e., a set of transformation rules with the properties of termination and confluence). 
We denote $KBA(E)$  as the result of the Knuth-Bendix algorithm  for $E$.
\end{definition}
\begin{example}
Let $A$ be an equation set s.t. 
\begin{eqnarray*}
A&=&
\begin{Bmatrix}
aa = \lambda,& baba = abab, &dbd = bdb,\\
 bb =  \lambda,&dbabd = abab, & da = ad, \\
dd = \lambda& & 
\end{Bmatrix}.
\end{eqnarray*}
We can compute $KBA(A)$:
\begin{eqnarray*}
KBA(A)&=&
\begin{Bmatrix}
aa \to \lambda,& ababd\to dbab,\\
abadb\to bdba,& abdba\to badb,\\
adbab\to babd,& baba\to abab,\\
babdb\to adba,& badba\to abdb,\\
bb \to \lambda,& bdbab\to abad,\\
da\to ad,& dbabd\to abab,\\
dbad\to bdba,& dbd\to bdb,\\
 dd \to \lambda
\end{Bmatrix}.
\end{eqnarray*}
\end{example}
Next, we apply the Knuth-Bedix completion algorithm  to 18 equations 
\begin{equation}\label{e0}
E_{all} =
\begin{Bmatrix}
aa=\lambda,& fbfb=a,& ab=ba\\
bb=\lambda,& adad=b,& bd=db\\
cc=\lambda,& dede=c,& cd=dc\\
dd=\lambda,& bcbc=d.& ce=ec\\
ee=\lambda,& fcfc=e,& af=fa\\
ff=\lambda,& eaea=f, & ef = fe
\end{Bmatrix}
\end{equation}
introduced by Iwama et. al. 2002 \cite{Iwama}.
We note that anti-commutative equations $xy = yxz$  $(x, y$ and $z \in \Sigma)$ equivalent to $xyxy = z$ $(xyxy = xyyxz = z)$.
We also call $xyxy = z$ $(x, y$ and $z \in \Sigma)$ anti-commutative equations. 
We used the $Mathmatica$ software to compute the complete transformation rule set $KBA(E_{all})$, and we list it in the Appendix.
The number of elements of $KBA(E_{all})$ is 114.
$$|KBA(E_{all})| = 114.$$
We note that we have applied an extended Knuth-Bendix algorithm which produce irreducible transformation rule set introduced in \cite{YM83}.
The transformation rule set $KBA(E_{all})$ is irreducible transformation rule set.
The number of rules obtained by the original Knuth-Bendix algorithm is 244.
A string rewriting system $(\Sigma, R_{E_{all}})$ is thus defined where $R_{E_{all}} = KBA(E_{all})$.

We would like to investigate commutativity of $E_{all}$.
%我々は交換子を計算することで回路の変形をより簡単に行えるようにする。
%Since $x^{-1} = x$, $y^{-1} = y$ for any $x$ and $y \in \Sigma$, we have $x^{-1}y^{-1}xy = xyxy$.
%Therefore, we compute $xyxy$.
%Let be a anti-commutative equation $xyxy = z$ $x, y$ and $z \in \Sigma$.  We can compute the commutator $yxyx$.
%$$yxyx = y^{-1}x^{-1}y^{-1}x^{-1} = (xyxy)^{-1} = z^{-1} = z.$$
%\begin{table}[!h]  \centering
%\hspace{.cm}
%\begin{tabular}{r|rrrrrr}
%$E_{all}$ & a & b & c & d & e & f\\
%\hline
%   a & $\lambda$ & $\lambda$ & * & b & f & $\lambda$ \\
%    b & $\lambda$ & $\lambda$ & d & $\lambda$ & * & a\\
%    c & * & d & $\lambda$ & $\lambda$ & $\lambda$ & e\\
%    d & b & $\lambda$ & $\lambda$ & $\lambda$ & c & *\\
%    e & f & * & $\lambda$ & c & $\lambda$ & $\lambda$\\
%    f & $\lambda$ & a & e & * & $\lambda$ & $\lambda$
%\end{tabular}
%\caption{Commutator $xyx^{-1}y^{-1}$ of $E_{all}$}\label{commutatorall}
%\end{table}
\begin{lemma}\label{comlemma}  We prove the following equations.
\begin{enumerate}
\item $(acac, ca) \in R_{E_{all}}$,
\item $(bebe, eb) \in R_{E_{all}}$ \mbox{and}
\item $(dfdf, fd) \in R_{E_{all}}$.
\end{enumerate}
\begin{proof}\item 
\begin{enumerate}
\item First, we show $fbca = caed$.
\begin{eqnarray*}
f(bc)a &=& f(cbd)a\\
&=& cfebda\\
&=& cfebbad\\
&=& cffaed\\
&=& caed.
\end{eqnarray*}
Since $fbca = caed$ and $fbfb = a$, 
\begin{eqnarray*}
acac &=& (fbfb)cac\\
&=& fb(caed)c\\
&=& (caed)edc\\
&=& ca(eded)c\\
&=& cacc\\
&=& ca.
\end{eqnarray*}
\item We can prove $bebe = eb$ by the same method to prove $acac = ca$.
We rewrite $a \to b$, $b \to d$, $c \to e$, $d \to c$, $e \to f$ and $f \to a$ in the proof of $acac = ca$.
\item We can prove $dfdf = fd$ by the same method to prove $acac = ca$.
We rewrite $a \to d$, $b \to c$, $c \to f$, $d \to e$, $e \to a$ and $f \to b$ in the proof of $acac = ca$.
\end{enumerate}
\end{proof}
\end{lemma}
Similarly, we obtain the following corollary.
\begin{corollary}\label{comcor}\item
\begin{enumerate}
\item $(caca, ac) \in R_{E_{all}}$,
\item $(ebeb, be) \in R_{E_{all}}$ \mbox{and}
\item $(fdfd, df) \in R_{E_{all}}$.
\end{enumerate}
\begin{proof}\item 
\begin{enumerate}
\item Since $acac = ca$ and $cc = \lambda$, we have 
\begin{eqnarray*}
caca &=& caca(cc)\\
&=& c(acac)c\\
&=& c(ca)c\\
&=& ac.
\end{eqnarray*}
\item 3. Similarly, we can prove.
\end{enumerate}
\end{proof}
\end{corollary}
By Lemma \ref{comlemma} and Corollary \ref{comcor}, we have the complete table of $xyxy$ for $\Sigma^*/E_{all}$.
\begin{table}[!h]  \centering
\hspace{.cm}
\begin{tabular}{|r|rrrrrr|}
\hline
\backslashbox{x}{y} & a & b & c & d & e & f\\
\hline
    a & $\lambda$ & $\lambda$ & ca & b & f & $\lambda$ \\
    b & $\lambda$ & $\lambda$ & d & $\lambda$ & eb & a\\
    c & ac & d & $\lambda$ & $\lambda$ & $\lambda$ & e\\
    d & b & $\lambda$ & $\lambda$ & $\lambda$ & c & fd\\
    e & f & be & $\lambda$ & c & $\lambda$ & $\lambda$\\
    f & $\lambda$ & a & e & df & $\lambda$ & $\lambda$\\
    \hline
\end{tabular}
\caption{$xyxy$ for $\Sigma^*/E_{all}$}\label{commutatorall}\label{com}
\end{table}
\begin{example}
We show an equation $(ebe, beb) \in \Sigma/R_{E_{all}}$. 
Since $ebe = ebe(bb) = (ebeb)b = be$, we have $(ebe, beb) \in \Sigma/R_{E_{all}}$.
We note that the rewriting rule $ebe \to beb$ appears on the last 6 line of Appendix.
\end{example}
\begin{proposition} 
Let $(\Sigma, R_{E_{all}})$ be a quantum circuit rewriting system where $E_{all}$ a set of equations defined by \mbox{(\ref{e0})}.
Then we have followings;
\begin{enumerate}
\item $|NF(w)| \le 6, (w \in \Sigma^7),$
\item $|NF(w)| \le 6, (w \in \Sigma^*)$, and
\item $|\Sigma^* / R_{E_{all}}| = 168.$
\end{enumerate}
That is the length of $NF(w)$ is at most 6 for any string $w \in \Sigma^*$ and the number of normal forms is 168.
\begin{proof}\item
\begin{enumerate}
\item We compute the $normal form$ for any string $w \in \Sigma^7$, then we have the length of $normal form$ is at most 6.
\item  For any string $w \in \Sigma^*$ which length is $n \ge 7$, $w$ contain a substring which length is 7.
Thus $w$ is rewritten to $w'$ which length is at most $n-1$.
Inductively, for all string $w \in \Sigma^*$, the length of $NF(w)$ is at most 6.
\item  We compute the $normal form$ for any string $w \in \Sigma^k$ $(1 \le k \le 6)$.
So we have all members of $\Sigma^*/R_{E_{all}}$ and we have $|\Sigma^*/R_{E_{all}}|=168$.
\end{enumerate}
\end{proof}
\end{proposition}
We list the all members of $\Sigma^*/R_{E_{all}}$ in Appendix.
The question now arises: Is the set of equations wordy?
Let $E_6$ be a set of equations such that
\begin{equation}\label{e1}
E_6 = E_{all} - 
\begin{Bmatrix}
ab=ba\\
bd=db\\
cd=dc\\
ce=ec\\
af=fa\\
ef = fe
\end{Bmatrix}
=
\begin{Bmatrix}
aa=\lambda,& fbfb=a\\
bb=\lambda,& adad=b\\
cc= \lambda,& dede=c\\
dd=\lambda,& bcbc=d\\
ee=\lambda,& fcfc=e\\
ff=\lambda,   & eaea=f
\end{Bmatrix}.
\end{equation}
The size of this equation set is $|E_6| = 12$. 
\begin{lemma}
We prove the following equations.
\begin{enumerate}
\item $(ba, ab) \in R_{E_6}$,
\item $(db, bd) \in R_{E_6}$,
\item $(dc, cd) \in R_{E_6}$,
\item $(ec, ce) \in R_{E_6}$,
\item $(fa, af) \in R_{E_6}$, \mbox{and}
\item $(fe, ef) \in R_{E_6}.$
\end{enumerate}
\begin{proof}\item
\begin{enumerate}
\item Since $bab = (adad)a(adad) = a$,
we have $(ba, ab) \in R_{E_6}$.
\item 3. 4. 5. 6. We can prove similarly.
\end{enumerate}\end{proof}
\end{lemma}
We compute a complete transformation rule set $KBA(E_6)$ by using the Knuth-Bendix algorithm,
and we can have $KBA(E_6) = KBA(E_{all})$.
The above results means that commutative type equations is not required for the initial equation set.
We have next proposition.
\begin{proposition}
Let $(\Sigma, R)$ be a quantum circuit rewriting system, $E_{all}$ and $E_6$ sets of equations defined by (\ref{e0}) and (\ref{e1}),
\begin{equation*}
\Sigma^* /R_{E_6} = \Sigma^* /R_{E_{all}}.
\end{equation*}
\hfill $\qed$
\end{proposition}
In the following section, we reduce the size of equation set and show the existence of the minimal set of equations $E_{min}$ of $E_6$ that generates the isomorphic monoid $\Sigma^*/ R_{E_{min}} = \Sigma^*/R_{E_6}$. 
%%%%%%%%%%%%%%%%%%%%%%%%%%
%
%
%
%
%%%%%%%%%%%%%%%%%%%%%%%%%%
\section{Minimal set of equations} 
\begin{definition}[Minimal set of equations]
Let $E\subseteq \Sigma ^* \times \Sigma ^*$. A subset $E_{min}\subset E$ is a minimal equation set of $E$ if and only if  
\begin{itemize}
\item $\Sigma^*/ R_{E_{min}} = \Sigma^*/R_{E}$, and
\item If $\Sigma^*/ R_{E'} = \Sigma^*/R_{E}$ then $|E_{min}| \le |E'|$ for all $E' \subset E$.
\end{itemize}
\end{definition}
In this section, we investigate a minimal set of equation of $E_6$ such that $\Sigma^*/ R_{E_{min}}= \Sigma^* /R_{E_6}$.
We delete some equations from $E_6$ and prove that the factor monoids of the equations are isomorphic.
We follow the same line of thought as was used for the elementary Tietze transformation \cite{Book}.
We first prove the following proposition.
\begin{proposition}\label{mainprop1}
Let $(\Sigma, R)$ be a quantum circuit rewriting system, $E_6$ a set of equations defined by (\ref{e1}),
\begin{equation*}
E_5 =
\begin{Bmatrix}
aa=\lambda,& fbfb=a\\
bb=\lambda,& adad=b\\
cc= \lambda,& dede=c\\
dd=\lambda,& bcbc=d\\
ee=\lambda,& fcfc=e\\
          & eaea=f
\end{Bmatrix}, \mbox{and}
\end{equation*}
\begin{eqnarray}
E_2 =
\begin{Bmatrix}
aa=\lambda,& fbfb=a\\
bb=\lambda,& adad=b\\
& dede=c\\
& bcbc=d\\
& fcfc=e\\
          & eaea=f
\end{Bmatrix}.
\end{eqnarray}
Then we have followings:
\begin{enumerate}
\item $(efc, cf), ((fc)e, e(fc)) \in R_{E_5},$
\item $(ff , \lambda) \in R_{E_5},$
\item $\Sigma^*/ R_{E_5} =\Sigma ^*/ R_{E_6}, \mbox{\  and}$
\item $\Sigma^*/ R_{E_2} =\Sigma ^*/ R_{E_6}.$
\end{enumerate}
\end{proposition}
\begin{proof}
We prove this proposition in following procedures.
\begin{enumerate}
\item Since
\begin{eqnarray*}
cf &=& (ee)(aeaeeaea)cf(cc)\\
     &=& e(eaea)e(eaea)cfcc\\
     &=& efe(fcfc)c\\
     &=& efeec\\
     &=& efc,
\end{eqnarray*}
we have $(efc, cf) \in R_{E_5}$.\\
Since
\begin{eqnarray*}
fce &=& fc(fcfc)\\
      &=& efc,
\end{eqnarray*}
we have $((fc)e, e(fc)) \in R_{E_5}$.
\item Since
\begin{eqnarray*}
ff &=& f(cc)f\\
     &=& fc(efc)\\
     &=& (efc)fc\\
     &=& e(e)\\
     &=& \lambda,
\end{eqnarray*}
we have $(ff, \lambda) \in R_{E_5}$.
\item We show that $\Sigma^*/R_{E_5} = \Sigma^*/R_{E_6}$.
Since $[ff]_{E_5} = [\lambda]_{E_5}$, we have $\Sigma^* / R_{E_5} = \Sigma^* / R_{E_6}$.
\item Let $E_4$, $E_3$ and $E_2$ be sets of equations where
\begin{eqnarray*}
E_4 &=& E_6 - \{ee = \lambda,ff = \lambda\}, \\
E_3 &=& E_6 - \{dd = \lambda,ee = \lambda,ff = \lambda\}, and \\
E_2 &=& E_6 - \{cc = \lambda, dd = \lambda,ee = \lambda,ff = \lambda\} .
\end{eqnarray*}
Similarly, we have $(ee,\lambda)\in R_{E_4}$,
$(dd,\lambda)\in R_{E_3}$ and $(cc,\lambda)\in R_{E_2}$.
So we obtain
%Similarly, the equations of eliminated type can be further removed: 
\begin{eqnarray*}
\Sigma^*/R_{E_6} &=& \Sigma^* /R_{E_5} = \Sigma^* /R_{E_3} = \Sigma^* /R_{E_4}\\
&=& \Sigma^* /R_{E_2}.
\end{eqnarray*}
\end{enumerate}
\end{proof}
%%%%%%%%%%%%%%%%%%%%%%%%%%
%
%
%
%
%
%
%%%%%%%%%%%%%%%%%%%%%%%%%%
Next, we prove that $E_2$ is a minimal set of equations.
\begin{proposition}\label{mainprop2}
\item Let $E' \subset E_6$.
If $\Sigma^*/ R_{E'} = \Sigma^*/R_{E_6}$, then $|E'| \geq 8$.
\end{proposition}
\begin{proof}
We will prove this in two stpdf.
\begin{description}
\item[ step 1:] First, we prove that we cannot remove a anti-commutative equation from $E_6$ .\\
We define a set of equations 
\begin{equation*}
E_{anti} =
\begin{Bmatrix}
fbfb=a, adad=b, dede=c,\\
bcbc=d,fcfc=e, eaea=f 
\end{Bmatrix}.
\end{equation*}
Let $u \in E_{anti}$ and consider $u$ as $xyxy = z$  ($x,y,z \in \Sigma$).
We define $E_u$ as $E_u = E_6 - \{u\} = E_6 - \{xyxy=z\}$.
Then we can show $\Sigma^*/R_{E_u} \neq \Sigma^*/R_{E_6}$ as follows.
We consider a monoid $M = (\{0,1\}, \cdot , 0)$ where a binary operator $\cdot \subset \{0,1\} \times \{0,1 \} \to \{0,1\}$ is defined by Table \ref{dot1}, and a function $i : \Sigma^* \to M$ is defined as
$$i(\lambda) = i(k) = 0 ( \forall k \neq z) , \mbox{\ and \ }  i(z) = 1.$$
We consider a homomorphism $i^*$ and show that $i^*$ is an interpretation for $E_u$: \\
$$i^*(kk) = 0 \cdot  0 = 0 = i^*(\lambda),$$
$$i^*(zz) = 1 \cdot 1 = 0 = i^*(\lambda), \mbox{\ and}$$
\begin{eqnarray*}
i^*(mnmn) &=& i^*(mn)\cdot i^*(mn) \\
&=& 0\\
&=&  i^*(k), \forall m, n \in \Sigma, k \neq z.
\end{eqnarray*}
Since $x \not=z$ and $y \not=z$, then the value of $i^*(xyxy)$ is
$$i^*(xyxy) = i(0)\cdot i(0)\cdot i(0)\cdot i(0) = 0 \cdot 0 \cdot 0 \cdot 0 =  0.$$
Since $i^*(z)$ is
$$i^*(z) = i(z) = 1.$$
We have $i^*(xyxy) \neq i^*(z).$
By Lemma \ref{interpretation}, $[xyxy]_{E_u}  \neq [z]_{E_u}$.
On the other hand, it is obvious that $[xyxy]_{E_6}  = [z]_{E_6}$.
Therefore, 
$$\Sigma^*/R_{E_u} \neq \Sigma^*/R_{E_6}.$$
\begin{table}[!h]  \centering
\begin{tabular}{r|rr}
$\cdot$ & 0 & 1\\
\hline
    0 & 0 & 1 \\
    1 & 1 & 0
\end{tabular}
\caption{Definition of the binary operator $\cdot $}\label{dot1}
\label{cdot_def}
\end{table}
\item[step 2:]
Let $x \in \Sigma$ and $E_{x}$ a set of equations defined by
\begin{eqnarray*}
E_{x} &=&
\begin{Bmatrix}
xx=\lambda,& fbfb=a\\
           & adad=b\\
           & dede=c\\
           & bcbc=d\\
           & fcfc=e\\
           & eaea=f
\end{Bmatrix}.
\end{eqnarray*}
Then we can have $\Sigma^*/R_{E_x}\not=\Sigma^*/R_{E_6}$ as follows.
For example, if we consider $x = a$, we can prove $\Sigma^*/ R_{E_{a}} \neq \Sigma ^*/ R_{E_6}.$
Let $N = (\{0,1, 2\}, \cdot , 0)$ be a monoid where a binary operator $\cdot \subset \{0,1\} \times \{0,1 \} \to \{0,1\}$ is  defined by Table  \ref{dot2}, and let the function $j : \Sigma^* \to N$ be $j(\lambda) = j(a) = j(c) = 0, j(b) = j(e) = 1$ and $j(d) = j(f) = 2.$
We consider a homomorphism $j^*$ and show that $j^*$ is an interpretation for $E_{a}$:
$$j^*(aa) = j(a)\cdot j(a) = 0 = j^*(\lambda),$$
$$j^*(fbfb) = j^*(fb)\cdot j^*(fb) = 0\cdot 0 = 0 = j^*(a),$$
$$j^*(adad) = j^*(ad)\cdot j^*(ad) = 2\cdot 2 = 1 =  j^*(b),$$
$$j^*(dede) = j^*(de)\cdot j^*(de) = 0\cdot 0 = 0 =  j^*(c),$$
$$j^*(bcbc) = j^*(bc)\cdot j^*(bc) = 1\cdot 1 = 2 =  j^*(d),$$
$$j^*(fcfc) = j^*(fc)\cdot j^*(fc) = 2\cdot 2 = 1 =  j^*(e)  \mbox{and}$$
$$j^*(eaea) = j^*(ea)\cdot j^*(ea) = 1 \cdot 1 = 2 =  j^*(d).$$
The value of $j^*(bb)$ is
$$j^*(bb) = j(b)\cdot j(b) = 1 \cdot 1 =  2.$$
The value of $j^*(\lambda)$ is
$$j^*(\lambda) = j(\lambda) = 0.$$
Since $j^*(bb) \neq j^*(\lambda)$, we have $[bb]_{E_{a}}  \neq [\lambda]_{E_{a}}$ by Lemma \ref{interpretation} .
On the other hand, it is obvious that $[bb]_{E_6}  = [\lambda]_{E_6}$.
Therefore, $$\Sigma^*/R_{E_{a}} \neq \Sigma^*/R_{E_6}.$$
\begin{table}[!h]  \centering
\hspace{.cm}
\begin{tabular}{r|rrr}
$\cdot$& 0 & 1 & 2\\
\hline
    0 & 0 & 1 & 2 \\
    1 & 1 & 2 & 0 \\
    2 & 2 & 0 & 1
\end{tabular}
\caption{Definition of  the binary operator $\cdot$}\label{dot2}
\end{table}
\end{description}
Let $E' \subset E_6$ be a set of equations.
If $\Sigma^*/R_{E'} = \Sigma^*/R_{E_6}$, it contained at least six anti-commutative equations by step1 and at least two eliminated equations by step2. 
Therefore, if $\forall E' \subset E_6$ it holds that $\Sigma^* /R_{E'} = \Sigma^*/ R_{E_6}$, then
$$|E'| \geq 8.$$
\end{proof}
\begin{lemma}
Let $E_{ac}$ be a set of equations where
\begin{eqnarray*}
E_{ac} &=&
\begin{Bmatrix}
aa=\lambda,& fbfb=a\\
           & adad=b\\
cc=\lambda,& dede=c\\
           & bcbc=d\\
           & fcfc=e\\
           & eaea=f
\end{Bmatrix}.
\end{eqnarray*}
Then 
\begin{equation}
\Sigma^*/ R_{E_{ac}} \neq \Sigma ^*/ R_{E_6}.
\end{equation}
\begin{proof}
We show $[bb]_{E_{ac}} \neq [\lambda]_{E_{ac}}$.
Let $N = (\{0,1, 2\}, \cdot , 0)$ be a monoid where a binary operator $\cdot \subset \{0,1\} \times \{0,1 \} \to \{0,1\}$ is defined by Table  \ref{dot2}, and let the function $j : \Sigma^* \to N$ be $j(\lambda) = j(a) = j(c) = 0, j(b) = j(e) = 1$ and $j(d) = j(f) = 2.$ 
The function $j$ is the same function used in step 2 of Proposition \ref{mainprop2}.
We consider a homomorphism $j^*$ and $j^*$ is an interpretation for $E_{a}$.
We check that $j^*(cc) = j^*(\lambda)$.
$$j^*(cc) = j(c)\cdot j(c) = 0 = j^*(\lambda).$$
The value of $j^*(bb)$ is
$$j^*(bb) = j(b)\cdot j(b) = 1 \cdot 1 =  2.$$
The value of $j^*(\lambda)$ is
$$j^*(\lambda) = j(\lambda) = 0.$$
Since $j^*(bb) \neq j^*(\lambda)$, we have $[bb]_{E_{ac}}  \neq [\lambda]_{E_{ac}}$ by Lemma \ref{interpretation}.
On the other hand, it is obvious that $[bb]_{E_6}  = [\lambda]_{E_6}$.
Therefore, $$\Sigma^*/R_{E_{ac}} \neq \Sigma^*/R_{E_6}.$$
\end{proof}
\end{lemma}

From the above discussion, we derive the next theorem.
\begin{theorem}\label{main}
There exists a minimal equation set
$E_{min}$ of $E_6$ such that $|E_{min}| = 8$.
\begin{proof}
$E_2$ is a minimal equation of $E_6$ by Propositions \ref{mainprop1} and \ref{mainprop2}.
\end{proof}
\end{theorem}

Next, we show that there is an 8 element set of equations $F_2 \nsubseteq E_6$ such that $\Sigma^*/ R_{F_2} = \Sigma ^*/ R_{E_6}$.
%It will be shown that the combination of another equations exists. 
\begin{proposition}
Let $F_2 \nsubseteq E_6$ be a set of equations defined by
\begin{equation*}
F_2 =
\begin{Bmatrix}
aa=\lambda,& bfbf=a\\
bb = \lambda,&dada=b\\
 	& eded=c\\
           &cbcb=d\\
           &cfcf=e\\
           & aeae=f
\end{Bmatrix}.
\end{equation*}
Then $\Sigma^*/ R_{F_2} = \Sigma ^*/ R_{E_6}$.
\begin{proof}
Let $F_{anti}$ and $F_6$ be  sets of equations defined by
\begin{equation*}
F_{anti} =
\begin{Bmatrix}
bfbf=a, &dada=b,&eded=c\\
cbcb=d,&cfcf=e,&aeae=f
\end{Bmatrix},\ \mbox{and}
\end{equation*}
\begin{equation*}
F_6 =
\begin{Bmatrix}
aa=\lambda,& bfbf=a\\
bb = \lambda, &dada=b\\
cc = \lambda, & eded=c\\
dd=\lambda, &cbcb=d\\
ee=\lambda, &cfcf=e\\
ff=\lambda,   & aeae=f
\end{Bmatrix}.
\end{equation*}
We can prove 
\begin{equation}
\Sigma^* /R_{F_2} = \Sigma^* /R_{F_6} \label{prop6eq0}
\end{equation} by following the same method in Proposition 4.
Since we can prove $bfbf=a, dada=b, eded=c,cbcb=d, cfcf=e \mbox{\ and\ } aeae=f$ in $F_6$, we have
\begin{equation}
\Sigma^* /R_{F_6} = \Sigma^* /R_{E_6 \cup F_{anti}}. \label{prop6eq1}
\end{equation}
Similarly, we can have
\begin{equation}
\Sigma^* /R_{E_6 \cup F_{anti}} = \Sigma^* /R_{E_6}. \label{prop6eq2}
\end{equation}
By (\ref{prop6eq0}), (\ref{prop6eq1}), and (\ref{prop6eq2}), we have $\Sigma^* /R_{F_2} = \Sigma^* /R_{E_6}. $
\end{proof}
\end{proposition}
\section{Conclusion \& Future Work}
We considered rewriting systems in order to reduce the size of quantum circuits. 
We compute a set of complete transformation rules using the Knuth-Bendix algorithm. 
We discovered that the length of the normal form of $w \in \Sigma$ is at most 6 and the number of $|\Sigma^* / R_{E_2}|$ is 168.

We found a minimal equation set $E_2$ of the set of equations $E_6$ such that $|E_2| = 8$.
On the other hand, we could construct a set of 8 equations 
$F_2 \nsubseteq E_6$ such that  $\Sigma^*/R_{E_6} = \Sigma^*/R_F$.
At this time, we do not have any equation set $E$ such that $\Sigma^*/R_{E_6}=\Sigma^*/R_E$ and $|E|<8$.
The calculation of Knuth-Bendix algorithm for a smaller equation set is not always faster.
We may take more computation time for a smaller equation set.
The computation time of our implementation of Knuth-Bendix algorithms take 150 seconds for $E_{all}$, 330 seconds for $E_6$ and 1200 seconds for $E_5$.

In this paper, we restricted the size of qubits, so as an area of further work, we need to investigate about 4 or more qubits quantum circuits.
We tried to compute Knuth-Bendix algorithm for 4 qubit quantum circuit rewriting system, we could not obtain results of computations..
We assume the cause to be computer power or set of equations.   
\section*{Acknowledgements}
I would like to thank Professor Yoshihiro Mizoguchi for his valuable advice and encouragement during the course of this study. 
I am grateful to Professor Miguel A. Martin-Delgado for his helpful comments.
%\nocite{*}
%\bibliographystyle{alpha}
%\bibliography{stringrewriting}

\newcommand{\etalchar}[1]{$^{#1}$}

\section*{Appendix}

We show the ruslut KBA($E_{all}$) of Kunth-Bendix algorithm for $E_{all}$ and the monoid $\Sigma^* / R_{E_{all}}$. 

{\footnotesize
KBA($E_{all}$) = \{
$aa$ $\to$ $\lambda$,  $abcae$ $\to$ $caed$,  $abcfbd$ $\to$ $beabc$,  $abcfd$ $\to$ $fbde$,  $abd$ $\to$ $da$,  $abea$ $\to$ $bef$,  $abef$ $\to$ $bea$,  $abf$ $\to$ $fb$,  $acaed$ $\to$ $bcae$,  $acfbd$ $\to$ $eabc$,  $ada$ $\to$ $bd$,  $adfcb$ $\to$ $dfcd$,  $adfcd$ $\to$ $dfcb$,  $aea$ $\to$ $ef$,  $aef$ $\to$ $ea$,  $afb$ $\to$ $bf$,  $ba$ $\to$ $ab$,  $bb$ $\to$ $\lambda$,  $bcab$ $\to$ $cda$,  $bcaeb$ $\to$ $afcda$,  \
$bcaed$ $\to$ $acae$,  $bcb$ $\to$ $cd$,  $bcd$ $\to$ $cb$,  $bceab$ $\to$ $adefd$,  \
$bcead$ $\to$ $adefb$,  $bda$ $\to$ $ad$,  $bdea$ $\to$ $adef$,  $bdef$ $\to$ $adea$,  $bfb$ $\to$ $af$,  $bfcb$ $\to$ $afcd$,  $bfcd$ $\to$ $afcb$,  $cabc$ $\to$ $acda$,  $cabeb$ $\to$ $bebdf$,  $cabed$ $\to$ $acabe$,  $cac$ $\to$ $aca$,  
$cad$ $\to$ $bca$,  $caebd$ $\to$ $fcda$,  $caeda$ $\to$ $fbc$,  $cafc$ $\to$ $acea$,  $cbc$ $\to$ $bd$,  $cbd$ $\to$ $bc$,  $cbe$ $\to$ $bed$,  $cbfc$ $\to$ $adea$,  $cc$ $\to$ $\lambda$,  $cdaeb$ $\to$ $ebdf$,  $cdaf$ $\to$ $bcfb$,  $cde$ $\to$ $ed$,  $cdfc$ $\to$ $def$,  $ceabc$ $\to$ $acbfd$,  $cebd$ $\to$ $deb$,  $ced$ $\to$ $de$,  $cef$ $\to$ $fc$,  $cfbc$ $\to$ $aeda$,  $cfbde$ $\to$ $beabc$,  $cfc$ $\to$ $ef$,  $dab$ $\to$ $ad$,  $dac$ $\to$ $acb$,  $dad$ $\to$ $ab$,  $daebd$ $\to$ $abceb$,  $daed$ $\to$ $abce$,  $dafcb$ $\to$ $bdfcd$,  $dafcd$ $\to$ $bdfcb$,  
$db$ $\to$ $bd$,  $dc$ $\to$ $cd$,  $dd$ $\to$ $\lambda$,  $deab$ $\to$ $cead$,  $dead$ $\to$ $ceab$,  $debd$ $\to$ $ceb$,  $ded$ $\to$ $ce$,  $dfb$ $\to$ $adf$,  $dfcda$ $\to$ $aebdf$,  $eabca$ $\to$ $acafd$,  $eabce$ $\to$ $bcfbd$,  $eabcf$ $\to$ $beabc$,  
$eabe$ $\to$ $befb$,  $eac$ $\to$ $acf$,  $eade$ $\to$ $afcd$,  $eadf$ $\to$ $dfcb$,  
$eae$ $\to$ $af$,  $eaf$ $\to$ $ae$,  $ebc$ $\to$ $deb$,  $ebde$ $\to$ $bceb$,  
$ebdfc$ $\to$ $bebdf$,  $ebdfd$ $\to$ $aebdf$,  $ebe$ $\to$ $beb$,  $ebf$ $\to$ $aeb$,  $ec$ $\to$ $ce$,  $edae$ $\to$ $bcfb$,  $edaf$ $\to$ $cdae$,  $ede, 
cd$,  $edf$ $\to$ $dfc$,  $ee$ $\to$ $\lambda$,  $efbc$ $\to$ $cfbd$,  $efbd$ $\to$ $aeda$,  $efc$ $\to$ $cf$,  $fa$ $\to$ $af$,  $fbca$ $\to$ $caed$,  $fbce$ $\to$ $abcf$,  $fbcf$ $\to$ $abce$,  $fbdeb$ $\to$ $beabc$,  $fbdf$ $\to$ $dafd$, 
$fbe$ $\to$ $bea$,  $fbf$ $\to$ $ab$,  $fca$ $\to$ $cae$,  $fcbf$ $\to$ $ceab$,  
$fcdae$ $\to$ $aebdf$,  $fcdf$ $\to$ $defd$,  $fce$ $\to$ $cf$,  $fcf$ $\to$ $ce$,  $fda$ $\to$ $bfd$,  $fde$ $\to$ $cfd$,  $fdf$ $\to$ $dfd$,  $fe$ $\to$ $ef$,  $ff$ $\to$ $\lambda$ \} 
}
\par
{\footnotesize
$\Sigma^* / R_{E_{all}} =$
$\{\lambda$, $a$, $ab$, $abc$, $abca$, $abcaf$, $abcafd$, $abce$, $abcea$, 
$abceb$, $abcf$, $abcfb$, $abe$, $abeb$, $abebd$, $abebdf$, $abed$, 
$abeda$, $ac$, $aca$, $acab$, $acabe$, $acae$, $acaeb$, $acaf$, 
$acafd$, $acb$, $acbf$, $acbfd$, $acd$, $acda$, $acdae$, $acdf$, 
$acdfd$, $ace$, $acea$, $aceab$, $acead$, $aceb$, $acf$, $acfb$, 
$acfd$, $ad$, $ade$, $adea$, $adeb$, $adef$, $adefb$, $adefd$, $adf$, 
$adfc$, $adfd$, $ae$, $aeb$, $aebd$, $aebdf$, $aed$, $aeda$, $af$, 
$afc$, $afcb$, $afcd$, $afcda$, $afd$, $b$, $bc$, $bca$, $bcae$, 
$bcaf$, $bcafd$, $bce$, $bcea$, $bceb$, $bcf$, $bcfb$, $bcfbd$, 
$bcfd$, $bd$, $bde$, $bdeb$, $bdf$, $bdfc$, $bdfcb$, $bdfcd$, $bdfd$, 
$be$, $bea$, $beab$, $beabc$, $bead$, $beb$, $bebd$, $bebdf$, $bed$, 
$beda$, $bef$, $befb$, $befd$, $bf$, $bfc$, $bfd$, $c$, $ca$, $cab$, 
$cabe$, $cae$, $caeb$, $caed$, $caf$, $cafd$, $cb$, $cbf$, $cbfd$, 
$cd$, $cda$, $cdae$, $cdf$, $cdfd$, $ce$, $cea$, $ceab$, $cead$, 
$ceb$, $cf$, $cfb$, $cfbd$, $cfd$, $d$, $da$, $dae$, $daeb$, $daf$, 
$dafc$, $dafd$, $de$, $dea$, $deb$, $def$, $defb$, $defd$, $df$, 
$dfc$, $dfcb$, $dfcd$, $dfd$, $e$, $ea$, $eab$, $eabc$, $ead$, $eb$, 
$ebd$, $ebdf$, $ed$, $eda$, $ef$, $efb$, $efd$, $f$, $fb$, $fbc$, 
$fbd$, $fbde$, $fc$, $fcb$, $fcd$, $fcda$, $fd$\}.
}
%\begin{table}[!h]  \centering
%\hspace{.cm}
%\begin{tabular}{r|rrrrrr}
 %x$\to$& a & b & c & d & e & f\\
%\hline
 %   axa & a & b & aca & bd & fe & f \\
 %   bxb & a & b & dc & d & beb & af\\
 %   cxc & cac & db & c & d & e & ef\\
  %  dxd & ba & b & c & d & ce & dfd\\
   % exe & fa & ebe & c & cd & e & f\\
    %fef & a & ab & ec & fdf & e & f
%\end{tabular}
%\end{table}
\end{document}